\documentclass{llncs}
\usepackage{pf}

\title{Instruction Sequence Size Complexity of Parity}
\author{J.A. Bergstra \and C.A. Middelburg}
\institute{Informatics Institute, Faculty of Science, University of
           Amsterdam, \\
           Science Park~904, 1098~XH Amsterdam, the Netherlands \\
           \email{J.A.Bergstra@uva.nl,C.A.Middelburg@uva.nl}}

\begin{document}
\maketitle

\begin{abstract}
Each Boolean function can be computed by a single-pass instruction 
sequence that contains only instructions to set and get the content of 
Boolean registers, forward jump instructions, and a termination 
instruction.
Auxiliary Boolean registers are not necessary for this.
In the current paper, we show that, in the case of the parity functions, 
shorter instruction sequences are possible with the use of an auxiliary 
Boolean register in the presence of instructions to complement the 
content of auxiliary Boolean registers.
This result supports, in a setting where programs are instruction 
sequences acting on Boolean registers, a basic intuition behind the
storage of auxiliary data, namely the intuition that this makes possible 
a reduction of the size of a program.
\begin{classcode}
F.1.1, F.1.3.
\end{classcode}
\end{abstract}

\section{Introduction}
\label{sect-intro}

In~\cite{BM13a}, we presented an approach to computational complexity in 
which algorithmic problems are viewed as families of functions that 
consist of an $n$-ary Boolean function for each natural number $n$ and 
the complexity of such problems is assessed in terms of the length of 
finite single-pass instruction sequences acting on Boolean registers 
that compute the members of these families.
The instruction sequences concerned contain only instructions to set and 
get the content of Boolean registers, forward jump instructions, and a 
termination instruction.
Moreover, each Boolean register used serves as either input register, 
output register or auxiliary register.

Auxiliary Boolean registers are not needed to compute Boolean functions.
The question whether shorter instruction sequences are possible with the 
use of auxiliary Boolean registers was not answered in~\cite{BM13a}.
In the current paper, we show that, in the case of the parity functions, 
shorter instruction sequences are possible with the use of an auxiliary 
Boolean register provided the instruction set is extended with
instructions to complement the content of auxiliary Boolean registers.
The parity function of arity $n$ is the function from $\Bool^n$ to 
$\Bool$ whose value at $b_1,\ldots,b_n$ is $\True$ if and only if the 
number of $\True$'s in $b_1,\ldots,b_n$ is odd.

The parity functions have a well-known practical application.
If we append to a bit string the value of the parity function for this 
bit~string, the appended bit is called the parity bit.
After appending the parity bit, the number of times that $\True$ occurs 
is always even.
A test for this property is called a parity check.
Appending parity bits and performing parity checks are used in many
techniques to detect errors in transmission of binary data 
(for an overview, see e.g.~\cite[Part~I]{Sal05a}).

In theoretical computer science, the complexity of the parity functions
has been extensively studied 
(see e.g.~\cite{FSS84a,Has86a,Weg91a,BIS12a,IMP12a,Has14a}).
All these studies have been carried out in the setting of Boolean 
circuits.
Moreover, to our knowledge, none of these studies is concerned with an 
issue comparable to the one considered in this paper, namely the issue 
whether, in the case of computing the parity functions, shorter 
instruction sequences are possible with the use of auxiliary Boolean 
registers than without the use of auxiliary Boolean registers.

Our results concerning this issue can be paraphrased as follows: in the 
presence of instructions to complement the content of auxiliary Boolean 
registers, 
(i)~the parity function of arity $n$ can be computed by an instruction 
sequence without the use of an auxiliary Boolean register where the 
length of the instruction sequence is not more than $5n-2$ but more than 
$2n+3$ and 
(ii)~the parity function of arity $n$ can be computed by an instruction 
sequence with the use of one auxiliary Boolean register where the length 
of the instruction sequence is not more than $2n+3$.
Consequently, smaller instruction sequences are possible with the use of
an auxiliary Boolean register.

In our results, the presence of instructions to complement the content 
of auxiliary Boolean registers is assumed.
In~\cite{BM13a}, instruction sequences that contain these instructions 
were not considered.
Since the results from that paper, with the exception of one auxiliary 
result, are concerned with upper bounds of instruction sequence size 
complexity, they go through if instruction sequences may contain 
instructions to complement the content of auxiliary Boolean registers as 
well.
However, when doing some other kinds of complexity analysis, e.g.\ when 
proving lower bounds of instruction sequence size complexity, the 
instructions assumed to be present can matter.
In~\cite{BM15a}, we consider all instructions for Boolean registers that 
are possible in the setting in which the work presented in this paper is 
carried out and study the effect of their presence on instruction 
sequence size.

We give in this paper, in a setting where programs are instruction
sequences acting on Boolean registers, an example where the storage of
auxiliary data makes possible a reduction of program size.
We cannot find any other work supporting the intuitively evident fact 
that it must happen quite often that the storage of auxiliary data leads 
to a reduction of program size.

Since Boolean registers can be looked upon as program variables, the
work on the transformation of programs to structured programs by coding
previous flow of control into auxiliary Boolean variables that is
presented in~\cite{BJ66a,Coo67a} can be considered related work.
However, we consider that work only moderately related because it
concerns the increase of the degree of structuredness of programs
instead of the decrease of the size of programs.
All the work on program correctness in which auxiliary program
variables play a role, starting with~\cite{OG76a} (see~\cite{ABO09a} for
an overview), can be considered somewhat related as well.
However, we consider that work only loosely related because it
restricts the use of auxiliary program variables to the left-hand sides
of assignments whereas, in this paper, the use of auxiliary Boolean
registers is not restricted at all.

The work presented in this paper is carried out in the setting of \PGA\ 
(ProGram Algebra).
\PGA\ is an algebraic theory of single-pass instruction sequences that
was taken as the basis of an approach to the semantics of programming 
languages introduced in~\cite{BL02a}.
As a continuation of the work presented in~\cite{BL02a},
(i)~the notion of an instruction sequence was subjected to systematic
and precise analysis and
(ii)~issues relating to diverse subjects in computer science and 
computer engineering were rigorously investigated in the setting of 
\PGA.
The subjects concerned include programming language expressiveness, 
computability, computational complexity, algorithm efficiency, 
algorithmic equivalence of programs, program verification, program 
compactness, probabilistic programming, and micro-architecture.
For a comprehensive survey of a large part of this work, 
see~\cite{BM12b}.
An overview of all the work done to date in the setting of \PGA\ and 
some open questions originating from that work can be found at
http://instructionsequence.wordpress.com.

This paper is organized as follows.
First, we present the preliminaries on instruction sequences and 
complexity classes based on them that are needed in the rest of the 
paper (Section~\ref{sect-PGA}).
Next, we describe how the parity functions can be computed by
instruction sequences without the use of auxiliary Boolean registers and
with the use of auxiliary Boolean registers 
(Section~\ref{sect-computing-pf}).
Then, we show that the smaller lengths of the instruction sequences in
the latter case cannot be obtained without the use of auxiliary Boolean 
registers (Section~\ref{sect-shorter-with-auxbr}).
Finally, we make some concluding remarks (Section~\ref{sect-concl}).

The preliminaries to the work presented in this paper include excerpts 
from the preliminaries to the work presented in~\cite{BM13a}.
The preliminaries include a brief summary of \PGA.
A comprehensive introduction to \PGA, including examples, can among 
other things be found in~\cite{BM12b}.

\section{Preliminaries: Instruction Sequences and Complexity}
\label{sect-PGA}

In this section, we present a brief outline of \PGA\ (ProGram Algebra), 
the particular fragment and instantiation of it that is used in this 
paper, and the kind of complexity classes considered in this paper.
A mathematically precise treatment for the case without instructions to
complement the content of Boolean registers can be found 
in~\cite{BM13a}.

The starting-point of \PGA\ is the simple and appealing perception
of a sequential program as a single-pass instruction sequence, i.e.\ a
finite or infinite sequence of instructions of which each instruction is
executed at most once and can be dropped after it has been executed or
jumped over.

It is assumed that a fixed but arbitrary set $\BInstr$ of
\emph{basic instructions} has been given.
The intuition is that the execution of a basic instruction may modify a 
state and produces a reply at its completion.
The possible replies are $\False$ and $\True$.
The actual reply is generally state-dependent.
Therefore, successive executions of the same basic instruction may
produce different replies.
The set $\BInstr$ is the basis for the set of instructions that may 
occur in the instruction sequences considered in \PGA.
The elements of the latter set are called \emph{primitive instructions}.
There are five kinds of primitive instructions, which are listed below:
\begin{itemize}
\item
for each $a \in \BInstr$, a \emph{plain basic instruction} $a$;
\item
for each $a \in \BInstr$, a \emph{positive test instruction} $\ptst{a}$;
\item
for each $a \in \BInstr$, a \emph{negative test instruction} $\ntst{a}$;
\item
for each $l \in \Nat$, a \emph{forward jump instruction} $\fjmp{l}$;
\item
a \emph{termination instruction} $\halt$.
\end{itemize}
We write $\PInstr$ for the set of all primitive instructions.

On execution of an instruction sequence, these primitive instructions
have the following effects:
\begin{itemize}
\item
the effect of a positive test instruction $\ptst{a}$ is that basic
instruction $a$ is executed and execution proceeds with the next
primitive instruction if $\True$ is produced and otherwise the next
primitive instruction is skipped and execution proceeds with the
primitive instruction following the skipped one --- if there is no
primitive instruction to proceed with,
inaction occurs;
\item
the effect of a negative test instruction $\ntst{a}$ is the same as
the effect of $\ptst{a}$, but with the role of the value produced
reversed;
\item
the effect of a plain basic instruction $a$ is the same as the effect
of $\ptst{a}$, but execution always proceeds as if $\True$ is produced;
\item
the effect of a forward jump instruction $\fjmp{l}$ is that execution
proceeds with the $l$th next primitive instruction of the instruction
sequence concerned --- if $l$ equals $0$ or there is no primitive
instruction to proceed with, inaction occurs;
\item
the effect of the termination instruction $\halt$ is that execution
terminates.
\end{itemize}

To build terms, \PGA\ has a constant for each primitive instruction and 
two operators. 
These operators are: the binary concatenation operator ${} \conc {}$ and 
the unary repetition operator ${}\rep$.
We use the notation $\Conc{i = k}{n} P_i$, 
where $k \leq n$ and $P_k,\ldots,P_n$ are \PGA\ terms, 
for the \PGA\ term $P_k \conc \ldots \conc P_n$.

The instruction sequences that concern us in the remainder of this paper 
are the finite ones, i.e.\ the ones that can be denoted by closed \PGA\ 
terms in which the repetition operator does not occur. 
Moreover, the basic instructions that concern us are instructions to set 
and get the content of Boolean registers.
More precisely, we take the set
\begin{ldispl}
\set{\inbr{i}.\getbr \where i \in \Natpos} \union
\set{\outbr.\setbr{b} \where b \in \Bool}
\\ \;\; {} \union
\set{\auxbr{i}.\getbr \where i \in \Natpos} \union
\set{\auxbr{i}.\setbr{b} \where i \in \Natpos \Land b \in \Bool} \union
\set{\auxbr{i}.\negbr \where i \in \Natpos}
\end{ldispl}%
as the set $\BInstr$ of basic instructions.%
\footnote
{We write $\Natpos$ for the set $\set{n \in \Nat \where n \geq 1}$ of
positive natural numbers.}

Each basic instruction consists of two parts separated by a dot.
The part on the left-hand side of the dot plays the role of the name of 
a Boolean register and the part on the right-hand side of the dot plays 
the role of a command to be carried out on the named Boolean register.
The names are employed as follows:
\begin{itemize}
\item
for each $i \in \Natpos$,
$\inbr{i}$ serves as the name of the Boolean register that is used as 
$i$th input register in instruction sequences;
\item
$\outbr$ serves as the name of the Boolean register that is used as
output register in instruction sequences;
\item
for each $i \in \Natpos$,
$\auxbr{i}$ serves as the name of the Boolean register that is used as 
$i$th auxiliary register in instruction sequences.
\end{itemize}
On execution of a basic instruction, the commands have the following 
effects:
\begin{itemize}
\item
the effect of $\getbr$ is that nothing changes and the reply is the 
content of the named Boolean register;
\item
the effect of $\setbr{\False}$ is that the content of the named Boolean 
register becomes $\False$ and the reply is $\False$;
\item
the effect of $\setbr{\True}$ is that the content of the named Boolean 
register becomes $\True$ and the reply is $\True$;
\item
the effect of $\negbr$ is that the content of the named Boolean 
register is complemented and the reply is the complemented content.
\end{itemize}

We will write $\ISbr$ for the set of all instruction sequences that can 
be denoted by a closed \PGA\ term in which the repetition operator does 
not occur in the case that $\BInstr$ is taken as specified above.
For each $k \in \Nat$, we will write $\ARISbr{k}$ for the set of all 
instruction sequences from $\ISbr$ in which primitive instructions of 
the forms $\auxbr{i}.c$, $\ptst{\auxbr{i}.c}$ and $\ntst{\auxbr{i}.c}$ 
with $i > k$ do not occur.
Moreover, we will write $\psize(X)$, where $X \in \ISbr$, for the length 
of $X$.

$\ISbr$ is the set of all instruction sequences that matter to the kind 
of complexity classes which will be introduced below.
$\ISbrna$ is the set of all instruction sequences from $\ISbr$ in 
which no auxiliary registers are used.

Let $n \in \Nat$, let $\funct{f}{\Bool^n}{\Bool}$, and 
let $X \in \ISbr$.
Then $X$ \emph{computes} $f$ if there exists a $k \in \Nat$ such that, 
for all $b_1,\ldots,b_n \in \Bool$, on execution of $X$ in an 
environment with input registers $\inbr{1},\ldots,\inbr{n}$, output 
register $\outbr$, and auxiliary registers $\auxbr{1},\ldots,\auxbr{k}$,
if 
\begin{itemize} 
\item 
for each $i \in \set{1,\ldots,n}$, 
the content of register $\inbr{i}$ is $b_i$ when execution starts;
\item
the content of register $\outbr$ is $\False$ when execution starts; 
\item 
for each $i \in \set{1,\ldots,k}$, 
the content of register $\auxbr{i}$ is $\False$ when execution~starts;
\end{itemize}
then the content of register $\outbr$ is $f(b_1,\ldots,b_n)$ when 
execution terminates.

A \emph{Boolean function family} is an infinite sequence
$\indfam{f_n}{n \in \Nat}$ of functions, where $f_n$ is an $n$-ary
Boolean function for each $n \in \Nat$.

Let $\IS \subseteq \ISbr$ and 
$\FN \subseteq \set{h \where \funct{h}{\Nat}{\Nat}}$.
Then $\nuc{\IS}{\FN}$ is the class of all Boolean function families
$\indfam{f_n}{n \in \Nat}$ for which
there exists an $h \in \FN$ such that,
for all $n \in \Nat$, there exists an $X \in \IS$ such that
$X$ computes $f_n$ and $\psize(X) \leq h(n)$. 
We will use the notation $\nuc{\IS}{B(f(n))}$, where 
$\IS \subseteq \ISbr$ and $\funct{f}{\Nat}{\Nat}$, for 
$\nuc{\IS}
  {\set{h \where
        \funct{h}{\Nat}{\Nat} \Land
        \Exists{m \in \Nat}
         {(\Forall{n \in \Nat}{(n \geq m \Limpl h(n) \leq f(n))})}}}$.

In~\cite{BM13a}, it is proved that $\nuc{\ISbr}{\poly}$ coincides with
\PTpoly.%
\footnote
{As usual, $\poly$ stands for  
 $\set{h \where 
      \funct{h}{\Nat}{\Nat} \Land h \mathrm{\,is \,polynomial}}$.
}

\section{Computing Parity Functions by Instruction Sequences}
\label{sect-computing-pf}

In this section, we describe how the parity functions can be computed by 
instruction sequences in two distinct cases: (i) the case where they are
computed without the use of auxiliary Boolean registers and (ii) the 
case where they are computed with the use of a single auxiliary Boolean 
register.

The $n$-ary parity function $\funct{\PF{n}}{\Bool^n}{\Bool}$ is defined 
by
\begin{ldispl}
\PF{n}(b_1,\ldots,b_n) = \True \quad \mathrm{iff} \quad
\textrm{the number of } \True\textrm{'s in } b_1,\ldots,b_n
\textrm{ is odd.}
\end{ldispl}%
We write $\PFF$ for the Boolean function family 
$\indfam{\PF{n}}{n \in \Nat}$.

We begin with defining instruction sequences which are intended to 
compute the parity functions without the use of auxiliary Boolean 
registers.
We define instruction sequences $\PIS{0}{0}$ and $\PIS{0}{1}$ as 
follows:
\begin{ldispl}
\PIS{0}{0} \deq \halt\;, \qquad 
\PIS{0}{1} \deq 
\ptst{\inbr{1}.\getbr} \conc \outbr.\setbr{\True} \conc \halt
\end{ldispl}%
and we uniformly define instruction sequences $\PIS{0}{n}$ for 
$n \geq 2$ as follows:
\begin{ldispl}
\PIS{0}{n} \deq 
\ptst{\inbr{1}.\getbr} \conc 
\Conc{i=2}{n}
 \bigl(\fjmp{4} \conc \ptst{\inbr{i}.\getbr} \conc \fjmp{3} \conc 
       \fjmp{3} \conc \ntst{\inbr{i}.\getbr}\bigr) \conc
\outbr.\setbr{\True} \conc \halt\;.
\end{ldispl}%
The instruction sequences defined above compute the parity functions.
It is not immediately clear why this is the case.
The following remark can be of help in understanding it.
The contents of $\inbr{1},\ldots,\inbr{n}$ are read in that order by
means of test instructions from the (unfolded) instruction sequence.
This happens in such a way that a register is read by means of a
positive test instruction if the value of the parity function for the
prefix of the bit string that has been dealt with before is $0$ and a
register is read by means of a negative test instruction if the value
of the parity function for the prefix of the bit string that has been
dealt with before is $1$.
\begin{proposition}
\label{proposition-correctness-PARzero}
For each $n \in \Nat$, $\PIS{0}{n}$ computes $\PF{n}$.
\end{proposition}
\begin{proof}
For $n < 2$, the proof is trivial.
We prove that $\PIS{0}{n}$ computes $\PF{n}$ for all $n \geq 2$ by 
induction on $n$.
The basis step consists of proving that $\PIS{0}{2}$ computes $\PF{2}$.
This follows easily by an exhaustive case distinction over the contents 
of $\inbr{1}$ and $\inbr{2}$.
The inductive step is proved in the following way.
It follows directly from the induction hypothesis that, after the 
$\fjmp{4} \conc \ptst{\inbr{i}.\getbr} \conc \fjmp{3} \conc
 \fjmp{3} \conc \ntst{\inbr{i}.\getbr}$ has been executed $n$ times, 
execution proceeds with the next instruction if the number of $\True$'s 
in the contents of $\inbr{1},\ldots,\inbr{n}$ is odd and otherwise the 
next instruction is skipped and execution proceeds with the instruction 
following the skipped one.
From this, it follows easily by an exhaustive case distinction over the 
content of $\inbr{n{+}1}$ that $\PIS{0}{n{+}1}$ computes $\PF{n{+}1}$.
\qed
\end{proof}
Because the instruction sequences defined above compute the parity 
functions, we have a first result about the complexity of $\PFF$.
\begin{theorem}
\label{theorem-complexity-class-0-pos}
$\PFF \in \nuc{\ISbrna}{B(5 \mul n - 2)}$.
\end{theorem}
\begin{proof}
By simple calculations, we obtain that $\psize(\PIS{0}{0}) = 1$ and, for 
each $n > 0$, $\psize(\PIS{0}{n}) = 5 \mul n - 2$. 
From this and Proposition~\ref{proposition-correctness-PARzero}, it 
follows immedi\-ately that $\PFF \in \nuc{\ISbrna}{B(5 \mul n - 2)}$.
\qed
\end{proof}

We go on with defining instruction sequences which are intended to 
compute the parity functions with the use of a single auxiliary Boolean 
register.
We define an instruction sequence $\PIS{1}{0}$ as follows:
\pagebreak[2]
\begin{ldispl}
\PIS{1}{0} \deq \halt
\end{ldispl}%
and we uniformly define instruction sequences $\PIS{1}{n}$ for 
$n \geq 1$ as follows:
\begin{ldispl}
\PIS{1}{n} \deq 
\Conc{i=1}{n}
 \bigl(\ptst{\inbr{i}.\getbr} \conc \auxbr{1}.\negbr\bigr) \conc
\ptst{\auxbr{1}.\getbr} \conc \outbr.\setbr{\True} \conc \halt\;.
\end{ldispl}%
The instruction sequences defined above compute the parity functions as
well.
At each point of the execution of $\PIS{1}{n}$, $\auxbr{1}$ contains the
value of the parity function for the prefix of the bit string that has 
been dealt with at that point.
\begin{proposition}
\label{proposition-correctness-PARone}
For each $n \in \Nat$, $\PIS{1}{n}$ computes $\PF{n}$.
\end{proposition}
\begin{proof}
For $n < 1$, the proof is trivial.
We prove that $\PIS{1}{n}$ computes $\PF{n}$ for all $n \geq 1$ by 
induction on $n$.
The basis step consists of proving that $\PIS{1}{1}$ computes $\PF{1}$.
This follows easily by an exhaustive case distinction over the content 
of $\inbr{1}$.
The inductive step is proved in the following way. 
It follows directly from the induction hypothesis that, after the 
$\ptst{\inbr{i}.\getbr} \conc \auxbr{1}.\negbr$ has been executed $n$ 
times, the content of $\auxbr{1}$ is $\True$ if the number of $\True$'s 
in the contents of $\inbr{1},\ldots,\inbr{n}$ is odd and otherwise the 
content of $\auxbr{1}$ is $\False$.
From this, it follows easily by an exhaustive case distinction over the 
content of $\inbr{n{+}1}$ that $\PIS{1}{n{+}1}$ computes $\PF{n{+}1}$.
\qed
\end{proof}
Because the instruction sequences defined above compute the parity 
functions as well, we have a second result about the complexity of 
$\PFF$.
\begin{theorem}
\label{theorem-complexity-class-1-pos}
$\PFF \in \nuc{\ARISbr{1}}{B(2 \mul n + 3)}$.
\end{theorem}
\begin{proof}
By simple calculations, we obtain that $\psize(\PIS{1}{0}) = 1$ and, for 
each $n > 0$, $\psize(\PIS{1}{n}) = 2 \mul n + 3$. 
From this and Proposition~\ref{proposition-correctness-PARone}, it 
follows immedi\-ately that $\PFF \in \nuc{\ARISbr{1}}{B(2 \mul n + 3)}$.
\qed
\end{proof}

Theorems~\ref{theorem-complexity-class-0-pos} 
and~\ref{theorem-complexity-class-1-pos} give rise to the question 
whether $\PFF \notin \nuc{\ISbrna}{B(2 \mul n + 3)}$.
In Section~\ref{sect-shorter-with-auxbr}, this question will be answered 
in the affirmative.
It is still an open question whether there exist a $k \geq 1$ and an
$\funct{f}{\Nat}{\Nat}$ with $f(n) < 2 \mul n + 3$ for all $n > 0$ such
that $\PFF \in \nuc{\ARISbr{k}}{B(f(n))}$.

According to the view taken in~\cite{BM14a}, differences in the number
of auxiliary Boolean registers whose use contributes to computing the 
function at hand always go with algorithmic differences.
This view is supported by the instruction sequences $\PIS{0}{n}$ and 
$\PIS{1}{n}$: in addition to having different lengths, they express 
undeniably quite different algorithms to compute $\PF{n}$ (for $n > 1$). 

\section{Shorter Instruction Sequences with Auxiliary Registers}
\label{sect-shorter-with-auxbr}

In Section~\ref{sect-computing-pf}, we have described how the parity 
functions can be computed by instruction sequences without the use of 
auxiliary Boolean registers and with the use of auxiliary Boolean 
registers.
In the current section, we show that the smaller lengths of the 
instruction sequences in the latter case cannot be obtained without the 
use of auxiliary Boolean registers.
In other words, we show that 
$\PFF \notin \nuc{\ISbrna}{B(2 \mul n + 3)}$.

First, we introduce some notation that will be used in this section.
We write $\ol{b}$, where $b\in \Bool$, for the complement of $b$, i.e.\
$\ol{\False} = \True$ and $\ol{\True} = \False$; 
we write $\ol{f}$, where $f$ is an $n$-ary Boolean function, for the
unique $n$-ary Boolean function $g$ such that 
$g(b_1,\ldots,b_n) = \ol{f(b_1,\ldots,b_n)}$; and
we write $\ol{F}$, where $F$ is a Boolean function family 
$\indfam{f_n}{n \in \Nat}$, for the Boolean function family
$\indfam{\ol{f_n}}{n \in \Nat}$.

We know the following about the complexity of $\ol{\PFF}$.
\begin{proposition}
\label{proposition-complement}
For each $k \in \Nat$, $\PFF \in \nuc{\ISbrna}{B(k)}$ implies
$\ol{\PFF} \in \nuc{\ISbrna}{B(k)}$.
\end{proposition}
\begin{proof}
It is sufficient to prove for an arbitrary $n \in \Natpos$ that any 
instruction sequence from $\ISbrna$ that computes $\PF{n}$ can be 
transformed into one with the same length that computes $\ol{\PF{n}}$.
So let $n \in \Natpos$, and let $X \in \ISbrna$ be such that $X$ 
computes $\PF{n}$.
We distinguish two cases: $n$ is odd and $n$ is even.

If $n$ is odd, then 
$\ol{\PF{n}}(b_1,\ldots,b_n) = \PF{n}(\ol{b_1},\ldots,\ol{b_n})$.
This implies that $\PF{n}$ is computed by the instruction sequence 
from $\ISbrna$ obtained from $X$ by replacing, for each 
$i \in \set{1,\ldots,n}$, $\ptst{\inbr{i}.\getbr}$ by 
$\ntst{\inbr{i}.\getbr}$ and vice versa.

If $n$ is even, then 
$\ol{\PF{n}}(b_1,\ldots,b_n) = \PF{n}(\ol{b_1},b_2,\ldots,b_n)$.
This implies that $\PF{n}$ is computed by the instruction sequence 
from $\ISbrna$ obtained from $X$ by replacing 
$\ptst{\inbr{1}.\getbr}$ by $\ntst{\inbr{1}.\getbr}$ and vice versa.
\qed
\end{proof}

The following three lemmas bring us step by step to the main result of 
this section, namely $\PFF \notin \nuc{\ISbrna}{B(2 \mul n + 3)}$.
\begin{lemma}
\label{lemma-complexity-class-0-neg-a}
Let $X \in \ISbrna$ be such that $X$ computes $\PF{2}$.
Then $\psize(X) \geq 6$.
\end{lemma}
\begin{proof}
Let $X \in \ISbrna$ be such that $X$ computes $\PF{2}$.
The following observations can be made about $X$:
\begin{itemize}
\item[(1)]
for one $i \in \set{1,2}$, there must be at least one instruction of 
the form $\ptst{\inbr{i}.\getbr}$ or the form $\ntst{\inbr{i}.\getbr}$ 
in $X$ and, for the other $i \in \set{1,2}$, there must be at least 
two instructions of the form $\ptst{\inbr{i}.\getbr}$ or the form 
$\ntst{\inbr{i}.\getbr}$ in $X$ --- because otherwise the final content 
of $\outbr$ will not in all cases be dependent on the content of both 
$\inbr{1}$ and $\inbr{2}$;
\item[(2)]
there must be at least one occurrence of $\outbr.\setbr{\True}$ in $X$ 
and the last occurrence of $\outbr.\setbr{\True}$ in $X$ must precede an 
occurrence of $\halt$ --- because otherwise the final content of 
$\outbr$ will never be $\True$;
\item[(3)]
there must be at least two occurrences of $\halt$ in $X$ unless there 
occurs an instruction of the form $\fjmp{l}$ in $X$ whose effect is that 
the last occurrence of $\outbr.\setbr{\True}$ in $X$ is skipped --- 
because otherwise the final content of $\outbr$ will never be $\False$.
\end{itemize}
It follows immediately from these observations that $\psize(X) \geq 6$.
\qed
\end{proof}

In the proofs of the next two lemmas, we use the term \emph{test on} 
$\inbr{i}$ to refer to an instruction of the form 
$\ptst{\inbr{i}.\getbr}$ or the form $\ntst{\inbr{i}.\getbr}$, and 
we use the term \emph{test} to refer to an instruction that is a test on 
$\inbr{1}$ or a test on $\inbr{2}$.
\begin{lemma}
\label{lemma-complexity-class-0-neg-b}
Let $X \in \ISbrna$ be such that $X$ computes $\PF{2}$.
Then $\psize(X) \geq 7$.
\end{lemma}
\begin{proof}
Let $X \in \ISbrna$ be such that $X$ computes $\PF{2}$.
Then, by Lemma~\ref{lemma-complexity-class-0-neg-a}, $\psize(X) \geq 6$.
It remains to be proved that $\psize(X) \neq 6$.
Below this is proved by contradiction, using the following assumption 
with respect to $X$: 
\begin{itemize}
\item
there does not occur an instruction of the form $\fjmp{l}$ in $X$ whose 
effect is that the last occurrence of $\outbr.\setbr{\True}$ in $X$ is 
skipped.
\end{itemize}
This assumption can be made without loss of generality because, if it 
is not met, $X$ can be replaced by an instruction sequence from 
$\ISbrna$ of the same or smaller length by which it is met.

Assume that $\psize(X) = 6$, and 
suppose that $X = u_1 \conc \ldots \conc u_6$.
The first occurrence of $\halt$ must be preceded by at least one test on 
$\inbr{1}$ and one test on $\inbr{2}$, because otherwise the final 
content of $\outbr$ will in some cases not be dependent on the content 
of both $\inbr{1}$ and $\inbr{2}$.
From this and observations~(1), (2), and (3) from the proof of 
Lemma~\ref{lemma-complexity-class-0-neg-a}, it follows that either $u_3$ 
or $u_4$ must be $\halt$ and that $u_5$ must be $\outbr.\setbr{\True}$ 
and $u_6$ must be $\halt$.
However, if $u_3 \equiv \halt$, then termination can take place after 
performing only one test.
Because this means that the final content of $\outbr$ will still in some 
cases not be dependent on the content of both $\inbr{1}$ and $\inbr{2}$, 
it is impossible that $u_3 \equiv \halt$.
So, 
$X = u_1 \conc u_2 \conc u_3 \conc \halt \conc
 \outbr.\setbr{\True} \conc \halt$. 
Let $Y = u_1 \conc u_2 \conc u_3 \conc \outbr.\setbr{\True} \conc \halt$.
Then $\psize(Y) = 5$ and, because $X$ computes $\PF{2}$ and $u_3$ is a 
test by observation~(1) from the proof of 
Lemma~\ref{lemma-complexity-class-0-neg-a}, $Y$ computes $\ol{\PF{2}}$.
Hence, by Proposition~\ref{proposition-complement}, there exists a
$Z \in \ISbrna$ that computes $\PF{2}$ such that $\psize(Z) \leq 5$.
This contradicts Lemma~\ref{lemma-complexity-class-0-neg-a}.
\qed
\end{proof}

Lemma~\ref{lemma-complexity-class-0-neg-a} is used in the proof of 
Lemma~\ref{lemma-complexity-class-0-neg-b}.
Lemma~\ref{lemma-complexity-class-0-neg-b}, in its turn, is used in a 
similar way below in the proof of 
Lemma~\ref{lemma-complexity-class-0-neg-c}.

A remark in advance about the proof of 
Lemma~\ref{lemma-complexity-class-0-neg-c} is perhaps in order.
A proof of this lemma needs basically an extremely extensive case 
distinction.%
\footnote
{At bottom, there are in the order of $10^7$ different instruction
 sequences to consider.
 By the assumptions made at the beginning of the proof of the lemma
 this number can be reduced by a factor of $10$.
}
In the proof given below, the extent of the case distinction is strongly
reduced by using various apposite properties of the instruction 
sequences concerned.
The reduction, which is practically necessary, has led to a proof that 
might make the impression to be unstructured.
\begin{lemma}
\label{lemma-complexity-class-0-neg-c}
Let $X \in \ISbrna$ be such that $X$ computes $\PF{2}$.
Then $\psize(X) > 7$.
\end{lemma}
\begin{proof}
Let $X \in \ISbrna$ be such that $X$ computes $\PF{2}$.
Then, by Lemma~\ref{lemma-complexity-class-0-neg-b}, $\psize(X) \geq 7$.
It remains to be proved that $\psize(X) \neq 7$.
Below this is proved by contradiction, using the following assumptions 
with respect to $X$:
\begin{itemize}
\item 
the first instruction of $X$ is a test;
\item
the instructions $\fjmp{0}$ and $\fjmp{1}$ do not occur in $X$;
\item
there does not occur an instruction of the form $\fjmp{l}$ in $X$ whose 
effect is that the last occurrence of $\outbr.\setbr{\True}$ in $X$ is 
skipped.
\end{itemize}
These assumptions can be made without loss of generality because, for 
each of them, if it is not met, $X$ can be replaced by an instruction 
sequence from $\ISbrna$ of the same or smaller length by which it is 
met.

Assume that $\psize(X) = 7$, and 
suppose that $X = u_1 \conc \ldots \conc u_7$.
From observations~(1), (2), and~(3) from the proof of 
Lemma~\ref{lemma-complexity-class-0-neg-a}, it follows that either 
$u_5$ must be $\outbr.\setbr{\True}$ and $u_6$ must be $\halt$ or
$u_6$ must be $\outbr.\setbr{\True}$ and $u_7$ must be $\halt$.
However, in the former case $X$ can be replaced by a shorter instruction 
sequence from $\ISbrna$.
Because this contradicts Lemma~\ref{lemma-complexity-class-0-neg-b}, 
it is impossible that $u_5 \equiv \outbr.\setbr{\True}$ and 
$u_6 \equiv \halt$.
So 
$X = u_1 \conc \ldots \conc u_5 \conc \outbr.\setbr{\True} \conc \halt$.

Consider the case that $u_1 \equiv \ptst{\inbr{1}.\getbr}$.
Because the first occurrence of $\halt$ must be preceded by at least one 
test on $\inbr{1}$ and one test on $\inbr{2}$, $u_2$ must be either 
$\fjmp{2}$, $\fjmp{3}$, $\fjmp{4}$ or a test.
If $u_2 \equiv \fjmp{2}$, then $X$ can be replaced by an instruction 
sequence whose length is $6$, to wit
$\ntst{\inbr{1}.\getbr} \conc u_3 \conc u_4 \conc u_5 \conc
 \outbr.\setbr{\True} \conc \halt$.
Because this contradicts Lemma~\ref{lemma-complexity-class-0-neg-b},
it is impossible that $u_2 \equiv \fjmp{2}$.
If $u_2 \equiv \fjmp{4}$ and moreover $\inbr{1}$ contains $\True$, then 
the final content of $\outbr$ will not be dependent on the content of 
$\inbr{2}$.
Therefore, it is also impossible that $u_2 \equiv \fjmp{4}$.
The cases that $u_2 \equiv \fjmp{3}$ and $u_2$ is a test need more
extensive investigation.

Because the first occurrence of $\halt$ must be preceded by at least one 
test on $\inbr{1}$ and one test on $\inbr{1}$, $u_3$ must be either 
$\fjmp{2}$, $\fjmp{3}$ or a test if $u_2 \equiv \fjmp{3}$.
If $u_2 \equiv \fjmp{3}$ and $u_3 \equiv \fjmp{2}$, then $X$ can be 
replaced by an instruction sequence whose length is $4$, to wit
$u_4 \conc u_5 \conc \outbr.\setbr{\True} \conc \halt$.
Because this contradicts Lemma~\ref{lemma-complexity-class-0-neg-b},
it is impossible that $u_3 \equiv \fjmp{2}$ if $u_2 \equiv \fjmp{3}$.
If $u_2 \equiv \fjmp{3}$ and $u_3 \equiv \fjmp{3}$ and moreover 
$\inbr{1}$ contains $\False$, then the final content of $\outbr$ will 
not be dependent on the content of $\inbr{2}$.
Therefore, it is impossible that $u_3 \equiv \fjmp{3}$ if 
$u_2 \equiv \fjmp{3}$.
So $u_3$ must be a test if $u_2 \equiv \fjmp{3}$.
Because $\outbr.\setbr{\True}$ has to be executed if $\inbr{1}$ contains 
$\True$ and $\inbr{2}$ contains $\False$, 
$u_5 \equiv \ntst{\inbr{2}.\getbr}$ if $u_2 \equiv \fjmp{3}$.
Moreover, because there must be at least two occurrences of $\halt$ in 
$X$ and $u_3$ must be a test if $u_2 \equiv \fjmp{3}$, 
$u_4$ must be $\halt$ if $u_2 \equiv \fjmp{3}$.
So 
$X = \ptst{\inbr{1}.\getbr} \conc \fjmp{3} \conc u_3 \conc \halt \conc
 \ntst{\inbr{2}.\getbr} \conc \outbr.\setbr{\True} \conc \halt$
and $u_3$ must be a test if $u_2 \equiv \fjmp{3}$.
In the case that $u_2 \equiv \fjmp{3}$, the subcase that $u_3$ is a test 
needs more extensive investigation.

If $u_3$ is a test, it is either $\ptst{\inbr{1}.\getbr}$, 
$\ntst{\inbr{1}.\getbr}$, $\ptst{\inbr{2}.\getbr}$ or 
$\ntst{\inbr{2}.\getbr}$.
If $u_2 \equiv \fjmp{3}$ and $u_3 \equiv \ptst{\inbr{1}.\getbr}$, then 
the final content of $\outbr$ will be independent of the content of 
$\inbr{1}$.
If $u_2 \equiv \fjmp{3}$ and $u_3 \equiv \ntst{\inbr{1}.\getbr}$, then 
the final content of $\outbr$ will be independent of the content of 
$\inbr{2}$ if $\inbr{1}$ contains $\False$.
If $u_2 \equiv \fjmp{3}$ and either $u_3 \equiv \ptst{\inbr{2}.\getbr}$ 
or $u_3 \equiv \ntst{\inbr{2}.\getbr}$, then the final content of 
$\outbr$ will be wrong if $\inbr{1}$ contains $\False$ and $\inbr{2}$ 
contains $\True$.
Therefore, it is impossible that $u_3$ is a test if 
$u_2 \equiv \fjmp{3}$.
Because there are no more alternatives left for $u_3$ if 
$u_2 \equiv \fjmp{3}$, it is impossible that $u_2 \equiv \fjmp{3}$. 
The left-over case for $u_2$ is the case that $u_2$ is a test.
This case needs very extensive investigation.

If $u_2$ is a test, it is either $\ptst{\inbr{1}.\getbr}$, 
$\ntst{\inbr{1}.\getbr}$, $\ptst{\inbr{2}.\getbr}$ or 
$\ntst{\inbr{2}.\getbr}$.
If $u_2 \equiv \ptst{\inbr{1}.\getbr}$, then $X$ can be replaced by an
instruction sequence whose length is $5$, to wit 
$u_3 \conc u_4 \conc u_5 \conc \outbr.\setbr{\True} \conc \halt$.
Because this contradicts Lemma~\ref{lemma-complexity-class-0-neg-b},
it is impossible that $u_2 \equiv \ptst{\inbr{1}.\getbr}$.
If $u_2 \equiv \ntst{\inbr{1}.\getbr}$, then $X$ can be replaced by an
instruction sequence whose length is $6$, to wit 
$\ntst{\inbr{1}.\getbr} \conc u_3 \conc u_4 \conc u_5 \conc
 \outbr.\setbr{\True} \conc \halt$.
Because this contradicts Lemma~\ref{lemma-complexity-class-0-neg-b},
it is impossible that $u_2 \equiv \ntst{\inbr{1}.\getbr}$.
The specific cases that $u_2 \equiv \ptst{\inbr{2}.\getbr}$ and 
$u_2 \equiv \ntst{\inbr{2}.\getbr}$ need more extensive investigation.

The following will be used in both these cases.
If both $u_2$ and $u_3$ are tests, then either $u_4$ or $u_5$ must be 
$\halt$ by observation~(3) from the proof of
Lemma~\ref{lemma-complexity-class-0-neg-a}.
If $u_5 \equiv \halt$, then 
(a)~it is impossible that $u_4 \equiv \halt$, because otherwise the 
final content of $\outbr$ will be independent of the content of 
$\inbr{2}$ if $\inbr{1}$ contains $\False$, and
(b)~it is impossible that $u_4 \equiv \fjmp{2}$ or $u_4$ is a test, 
because otherwise $X$ can be replaced by an instruction sequence whose 
length is $6$ and this contradicts 
Lemma~\ref{lemma-complexity-class-0-neg-b}.
So it is impossible that $u_5 \equiv \halt$ if both $u_2$ and $u_3$ are 
tests and $u_4$ must be $\halt$ if both $u_2$ and $u_3$ are tests.
Moreover, if $u_2$ is a test, then it is impossible that 
$u_3 \equiv \halt$, because otherwise the final content of $\outbr$ will 
be independent of the content of $\inbr{2}$ if $\inbr{1}$ contains 
$\False$.

Because it is impossible that $u_3 \equiv \halt$ if $u_2$ is a test, 
$u_3$ must be either $\fjmp{2}$, $\fjmp{3}$ or a test if 
$u_2 \equiv \ptst{\inbr{2}.\getbr}$.
If $u_2 \equiv \ptst{\inbr{2}.\getbr}$ and $u_3 \equiv \fjmp{2}$, then 
the final content of $\outbr$ will be wrong if $\inbr{1}$ contains 
$\False$ unless $u_5 \equiv \ptst{\inbr{2}.\getbr}$.
However, then $u_4$ must be $\halt$ by observation~(3) from the proof of
Lemma~\ref{lemma-complexity-class-0-neg-a} and the final content of 
$\outbr$ will be wrong if $\inbr{1}$ contains $\True$ and $\inbr{2}$ 
contains $\True$.
Therefore, it is impossible that $u_3 \equiv \fjmp{2}$ if 
$u_2 \equiv \ptst{\inbr{2}.\getbr}$. 
If $u_2 \equiv \ptst{\inbr{2}.\getbr}$ and $u_3 \equiv \fjmp{3}$, then 
the final content of $\outbr$ will be independent of the content of
$\inbr{2}$ if $\inbr{1}$ contains $\False$.
Therefore, it is impossible that $u_3 \equiv \fjmp{3}$ if 
$u_2 \equiv \ptst{\inbr{2}.\getbr}$. 
So $u_3$ must be a test if $u_2 \equiv \ptst{\inbr{2}.\getbr}$.
We know that $u_4$ must be $\halt$ and it is impossible that 
$u_5 \equiv \halt$ if both $u_2$ and $u_3$ are tests. 
This implies that, if $u_2 \equiv \ptst{\inbr{2}.\getbr}$ and $u_3$ is 
a test, the final content of $\outbr$ will be wrong if $\inbr{1}$ 
contains $\True$ and $\inbr{2}$ contains $\False$.
Therefore, it is impossible that $u_3$ is a test if 
$u_2 \equiv \ptst{\inbr{2}.\getbr}$. 
Because there are no more alternatives left for $u_3$ if 
$u_2 \equiv \ptst{\inbr{2}.\getbr}$, it is impossible that 
$u_2 \equiv \ptst{\inbr{2}.\getbr}$.

Because it is impossible that $u_3 \equiv \halt$ if $u_2$ is a test, 
$u_3$ must be either $\fjmp{2}$, $\fjmp{3}$ or a test if 
$u_2 \equiv \ntst{\inbr{2}.\getbr}$.
If $u_2 \equiv \ntst{\inbr{2}.\getbr}$ and $u_3 \equiv \fjmp{2}$, then 
the final content of $\outbr$ will be wrong if $\inbr{1}$ contains 
$\True$ and $\inbr{2}$ contains $\False$ unless 
$u_5 \equiv \ntst{\inbr{2}.\getbr}$.
However, then $u_4$ must be $\halt$ by observation~(3) from the proof 
of Lemma~\ref{lemma-complexity-class-0-neg-a} and the final content of 
$\outbr$ will be wrong if $\inbr{1}$ contains $\False$ and $\inbr{2}$ 
contains $\True$.
Therefore, it is impossible that $u_3 \equiv \fjmp{2}$ if 
$u_2 \equiv \ntst{\inbr{2}.\getbr}$. 
If $u_2 \equiv \ntst{\inbr{2}.\getbr}$ and $u_3 \equiv \fjmp{3}$, then 
the final content of $\outbr$ will be independent of the content of 
$\inbr{2}$ if $\inbr{1}$ contains $\False$.
Therefore, it is impossible that $u_3 \equiv \fjmp{3}$ if 
$u_2 \equiv \ntst{\inbr{2}.\getbr}$. 
So $u_3$ must be a test if $u_2 \equiv \ntst{\inbr{2}.\getbr}$.
In the case that $u_2 \equiv \ntst{\inbr{2}.\getbr}$, the subcase that 
$u_3$ is a test needs more extensive investigation.

If $u_3$ is a test, it is either $\ptst{\inbr{1}.\getbr}$, 
$\ntst{\inbr{1}.\getbr}$, $\ptst{\inbr{2}.\getbr}$ or 
$\ntst{\inbr{2}.\getbr}$.
We know that $u_4$ must be $\halt$ and it is impossible that 
$u_5 \equiv \halt$ if both $u_2$ and $u_3$ are tests. 
This implies that 
(a)~if $u_2 \equiv \ntst{\inbr{2}.\getbr}$ and either
$u_3 \equiv \ptst{\inbr{1}.\getbr}$ or
$u_3 \equiv \ntst{\inbr{2}.\getbr}$, then the final content of $\outbr$ 
will be wrong if $\inbr{1}$ contains $\True$ and $\inbr{2}$ contains 
$\False$ and
(b)~if $u_2 \equiv \ntst{\inbr{2}.\getbr}$ and either 
$u_3 \equiv \ntst{\inbr{1}.\getbr}$ or 
$u_3 \equiv \ptst{\inbr{2}.\getbr}$, then the final content of $\outbr$ 
will be wrong if $\inbr{1}$ contains $\False$ and $\inbr{2}$ contains 
$\True$.
Because there are no more alternatives left for $u_3$ if 
$u_2 \equiv \ntst{\inbr{2}.\getbr}$, it is impossible that 
$u_2 \equiv \ntst{\inbr{2}.\getbr}$.

Because there are no more alternatives left for $u_2$, it is impossible 
that $u_1 \equiv \ptst{\inbr{1}.\getbr}$.
Analogously, we find that 
it is impossible that $u_1 \equiv \ntst{\inbr{1}.\getbr}$,
it is impossible that $u_1 \equiv \ptst{\inbr{2}.\getbr}$, and
it is impossible that $u_1 \equiv \ntst{\inbr{2}.\getbr}$.
Hence, it is impossible that $\psize(X) = 7$.
\qed
\end{proof}

Lemma~\ref{lemma-complexity-class-0-neg-b} is used in the proof of 
Lemma~\ref{lemma-complexity-class-0-neg-c}.
The latter lemma, in its turn, is used below in the proof of 
the main result of this section.
\begin{theorem}
\label{theorem-complexity-class-0-neg}
$\PFF \notin \nuc{\ISbrna}{B(2 \mul n + 3)}$.
\end{theorem}
\begin{proof}
It is sufficient to prove that there exists an $m \in \Natpos$ such 
that, for all $n \geq m$, for all instruction sequences $X \in \ISbrna$ 
that compute $\PF{n}$, $\psize(X) > 2 \mul n + 3$.
We take $m = 2$ because of Lemma~\ref{lemma-complexity-class-0-neg-c} 
and prove the property to be proved for all $n \geq 2$ by induction on 
$n$.
The basis step consists of proving that, for all instruction sequences 
$X \in \ISbrna$ that compute $\PF{2}$, $\psize(X) > 7$. 
This follows trivially from Lemma~\ref{lemma-complexity-class-0-neg-c}. 
The inductive step is proved below by contradiction.

Suppose that $X \in \ISbrna$, $X$ computes $\PF{n+1}$, and
$\psize(X) \leq 2 \mul (n + 1) + 3$.
Without loss of generality, we may assume that there does not exist an 
$X' \in \ISbrna$ that computes $\PF{n+1}$ such that 
$\psize(X') < \psize(X)$.
Therefore, we may assume that $X = u_1 \conc \ldots \conc u_k$ 
($ k \leq 2 \mul (n + 1) + 3$) where 
$u_1 \equiv \ptst{\inbr{n{+}1}.\getbr}$ or
$u_1 \equiv \ntst{\inbr{n{+}1}.\getbr}$.
We distinguish these two cases.

In the case that $u_1 \equiv \ntst{\inbr{n{+}1}.\getbr}$, we consider 
the case that $\inbr{n{+}1}$ contains $\True$.
In this case, after execution of $u_1$, execution proceeds with $u_3$.
Let $Y \in \ISbrna$ be obtained from $u_3 \conc \ldots \conc u_k$ by 
first replacing $\ptst{\inbr{n{+}1}.\getbr}$ by $\fjmp{1}$ and
$\ntst{\inbr{n{+}1}.\getbr}$ by $\fjmp{2}$ and then removing the 
$\fjmp{1}$'s.
Then, we have that $Y$ computes $\ol{\PF{n}}$ and
$\psize(Y) \leq \psize(X) - 2 \leq 2 \mul n + 3$.
Hence, by Proposition~\ref{proposition-complement}, there exists a 
$Z \in \ISbrna$ that computes $\PF{n}$ such that 
$\psize(Z) \leq 2 \mul n + 3$.
This contradicts the induction hypothesis.

In the case that $u_1 \equiv \ptst{\inbr{n{+}1}.\getbr}$, we consider 
the case that $\inbr{n{+}1}$ contains $\False$.
This case leads to a contradiction in the same way as above, but  
without the use of Proposition~\ref{proposition-complement}.
\qed
\end{proof}

The following is a corollary of 
Theorems~\ref{theorem-complexity-class-1-pos}
and~\ref{theorem-complexity-class-0-neg}.
\begin{corollary}
$\nuc{\ISbrna}{B(2 \mul n + 3)} \subset
 \nuc{\ARISbr{1}}{B(2 \mul n + 3)}$.%
\footnote
{As usual, the symbol $\subset$ is used to denote the proper inclusion
 relation.}
\end{corollary}

\section{Concluding Remarks}
\label{sect-concl}

We have shown, in a setting where programs are instruction sequences 
acting on Boolean registers, that, in the case of the parity functions,
smaller programs are possible with the use of one auxiliary Boolean 
register than without the use of an auxiliary Boolean register.
This result supports a basic intuition behind the storage of 
auxiliary data, namely the intuition that this makes possible a 
reduction of the size of a program.

Of course, more results supporting this intuition would be nice. 
Adversely, we do not even know at the present time how to prove, for 
example, that there exists a Boolean function family for which smaller 
instruction sequences are possible with the use of two auxiliary Boolean 
registers than with the use of one auxiliary Boolean register.
Moreover, we do not know of results in other theoretical settings that 
support the intuition that the storage of auxiliary data makes possible 
a reduction of the size of a program.

It is still an open question whether smaller programs are possible with 
the use of one auxiliary Boolean register than without the use of an 
auxiliary Boolean register in the absence of instructions to complement
the content of auxiliary Boolean registers. 
We conjecture that this question can be answered in the affirmative, but
the practical problem is that the number of different instruction 
sequences to be considered in the proof is increased by a factor of 
$10^4$.

It is intuitively clear that the instructions sequences $\PIS{0}{n}$ and 
$\PIS{1}{n}$ from Section~\ref{sect-computing-pf} express parameterized 
algorithms.
However, it is very difficult to develop a precise viewpoint on what is 
a parameterized algorithm.
In~\cite{BM14a}, we looked for an equivalence relation on instruction 
sequences that captures to a reasonable degree the intuitive notion that 
two instruction sequences express the same algorithm.
In that paper, we considered non-parameterized algorithms only because 
it turned out to be already very difficult to develop a precise 
viewpoint on what is a non-parameterized algorithm.

\subsection*{Acknowledgements}
We thank two anonymous referees for carefully reading a preliminary 
version of this paper and for suggesting improvements of the 
presentation of the paper.

\bibliographystyle{splncs03}
\bibliography{IS}

\end{document}